%% file: main.tex
\documentclass[11pt]{article}
\usepackage{amssymb,a4wide}
\usepackage{epsfig}
\usepackage{enumerate}
\usepackage{color,enumerate}
\newcommand{\comment}[1]{}

\usepackage{amsmath}
\usepackage{amssymb}
\usepackage{amsthm}
\usepackage{tikz}
\usepackage[lined,algonl,boxed]{algorithm2e}

\usepackage{graphicx}

\newcommand{\URL}[1]{\textcolor{blue}{\url{#1}}}

\newtheorem{theorem}{Theorem}
\newtheorem{corollary}[theorem]{Corollary}
\newtheorem{proposition}[theorem]{Proposition}

\theoremstyle{definition}
\newtheorem{definition}{Definition}

\usepackage[langfont=caps]{complexity}
\newclass{\LFew}{LFew}
\newclass{\FewL}{FewL}
\newclass{\FL}{FL}
\newclass{\OptL}{OptL}
\newclass{\FewUL}{FewUL}
\newclass{\promiseUL}{promiseUL}
\newclass{\UOptL}{UOptL}
\newclass{\uspace}{Uspace}
\newclass{\nspace}{Nspace}
\newclass{\amb}{Ambiguity}
\newclass{\ReachFewL}{ReachFewL}
\newclass{\ReachLFew}{ReachLFew}
\newclass{\uncongested}{UnCongested}
\newclass{\coUL}{coUL}
\newclass{\ReachUL}{ReachUL}
\newclass{\LogDCFL}{LogDCFL}
\newlang{\GGR}{GGR}
\newlang{\MGGR}{3D-mGGR}
\newlang{\gB}{Bi}
\newlang{\gBpoly}{fp-Bi}
\newlang{\gGG}{GG}
\newclass{\page}{3Page}
\newlang{\pageR}{3PageReach}
\numberwithin{equation}{section}
\newcommand{\acc}{\mbox{\it acc}}
\newcommand{\rej}{\mbox{\it rej}}
\newcommand{\gap}{\mbox{\it gap}}

\title{\bf\LARGE On the Power of Unambiguity in
  Logspace\footnote{Research supported in part by NSF grants
    CCF-0830730, CCF-0916525, CCF-0830479, CCF-0916797.}}

\author{ {\sc Aduri Pavan}\footnote{Department of Computer Science,
    Iowa State University: {\tt email:pavan@cs.iastate.edu}} \and {\sc
    Raghunath Tewari}\footnote{Department of Computer Science and
    Engineering, University of Nebraska-Lincoln: {\tt
      email:rtewari@cse.unl.edu}} \and {\sc
    N. V. Vinodchandran}\footnote{Department of Computer Science and
    Engineering, University of Nebraska-Lincoln: {\tt
      email:vinod@cse.unl.edu}}}

\begin{document}

\maketitle

\vspace{2cm}
\begin{abstract}
We report progress on the \NL\ vs \UL\ problem. 
\begin{itemize}
\item[-] We show unconditionally that the complexity class $\ReachFewL\subseteq\UL$. This improves on the earlier known upper bound $\ReachFewL \subseteq \FewL$.
\item[-] We investigate the complexity of min-uniqueness - a central
  notion in studying the \NL\ vs \UL\ problem.
\begin{itemize} 
\item We show that min-uniqueness is necessary and sufficient for  showing 
$\NL\ =\UL$. 
\item We revisit the class $\OptL[\log n]$ and show that {\sc ShortestPathLength} - computing the length of the shortest path in a DAG, is complete for $\OptL[\log n]$. 
\item We introduce $\UOptL[\log n]$, an unambiguous version of $\OptL[\log n]$, and show that (a) $\NL =\UL$ if and only if $\OptL[\log n] = \UOptL[\log n]$, (b) $\LogFew \leq \UOptL[\log n] \leq \SPL$.  
\end{itemize}
\item[-] We show that the reachability problem over graphs embedded on 3 pages is complete for \NL. This contrasts with the reachability problem over graphs embedded on 2 pages which is logspace equivalent to the reachability problem in planar graphs and hence is in \UL. 
\end{itemize}
\end{abstract}
\newpage
\input{intro}

\input{draft.tex}

\input{3page.tex}

\input{ack.tex}

\bibliographystyle{alpha}
\bibliography{GGR,data}

\end{document}

%% file: intro.tex
\section{Introduction}

This paper is centered around the $\NL$ vs $\UL$ problem.  Can
nondeterministic space bounded computations be made unambiguous? This
fundamental question was first raised by Reinhardt and Allender in the
paper entitled ``Making Nondeterminism Unambiguous''
\cite{ReinhardtAllender00}. Reinhardt and Allender showed that in the
non-uniform setting it is indeed possible to simulate any
nondeterministic logspace computation by an unambiguous one (that is,
$\NL/\poly = \UL/\poly$) thus giving the first strong evidence that
this relation might hold in the uniform setting as well.

A nondeterministic machine is unambiguous if it has at most one
accepting path on any input~\cite{Valiant76}. \UL\ is the class of
decision problems that are decided by unambiguous logspace bounded
nondeterministic machines. Clearly $\UL$ is the natural logspace
analog of $\UP$~\cite{Valiant76}, the unambiguous version of
$\NP$. Historically, several researchers have investigated this class
(for example, \cite{BHS93, AJ93, BJLR91, BDHM92}) in different
contexts. But Buntrock et al. ~\cite{BJLR91} are the first to conduct
a focused study of the complexity class $\UL$ and its variations. 


Since the above-mentioned paper due to Reinhardt and Allender, there
has been significant progress reported on the $\NL$ vs $\UL$ problem.
In \cite{AllenderMatching99}, Allender, Reinhard, and Zou showed that,
under the (very plausible) hardness assumption that deterministic
linear space has functions that can not be computed by circuits of
size $2^{\epsilon n}$, the constructions given by Reinhardt and
Allender can be {\em derandomized} to show that
$\NL=\UL$~\cite{AllenderMatching99}.  As the reachability problem for
directed graphs is complete for $\NL$, it is natural to investigate
the space complexity of reachability for subclasses of directed graphs
and indeed the recent progress has been in this direction. In
\cite{BTV09}, it is shown that reachability for directed planar graphs
is in $\UL$. Subsequently, Thierauf and Wagner showed that
reachability for $K_{3,3}$-free and $K_{5}$-free graphs can be reduced
to planar reachability in logspace~\cite{TW09}. Kyn\v{c}l and
Vysko\v{c}il showed that reachability for bounded genus graphs also
reduces to the planar case \cite{KV09}. Thus reachability for these
classes of graphs is also in $\UL$.  

These results provide significant evidence that $\NL$ equals $\UL$ and
establishing this fundamental equivalence may be within the reach of
current techniques.

\subsection*{Our Results}

\subsubsection*{Complexity of \ReachFewL}

$\FewL$, the logspace analog of the polynomial time class
$\FewP$~\cite{Allender86,CaiHemachandra90}, is the class of languages
that are decided by nondeterministic logpsace machines with the
promise that on any input there are at most polynomially many
accepting paths~\cite{BJLR91,BDHM92}.  Is $\FewL = \UL$? As $\FewL
\subseteq \NL$, this is a very interesting restriction of $\NL = \UL$
question (it is known that $\FewL$ is in $\L^{\promiseUL}$
\cite{Allender06}). While we are unable to show that $\FewL\subseteq
\UL$ , as our first result we show that the class $\ReachFewL
\subseteq \UL$.

\vspace{2mm}
\noindent {\em Result 1}. $\ReachFewL\subseteq \UL\cap\coUL$.
\vspace{1mm}

\ReachFewL\ is a restriction of $\FewL$~\cite{BJLR91}.  We call a
nondeterministic machine $M$ a {\em reach-few} machine, if for any
input $x$ and any configuration $c$ of $M(x)$, the number of paths
from the start configuration to $c$, is bounded by a
polynomial. $\ReachFewL$ is the class of languages decided by a
reach-few machine that is logspace bounded. Notice that for a machine
accepting a $\FewL$ language there can be (useless) configurations
which does not lead to any accepting configuration but still with
exponentially many paths from the start configuration to them. For a
reach-few machine, the number of paths from the start configuration to
{\em any} configuration is bounded by a polynomial. It is worth noting
that such distinctions are not meaningful in the polynomial time
setting as there is enough space to store the entire computation path
during a nondeterministic computation. Our result improves on the
previous known trivial upper bound of $\ReachFewL\subseteq \FewL$.

The class $\ReachFewL$ was also investigated by Buntrock, Hemachandra,
and Siefkes \cite{BHS93} under the notation \nspace-\amb$(\log
n,n^{O(1)})$. In \cite{BHS93}, the authors define, for a space bound $s$
and an unambiguity parameter $a$, the class
\nspace-\amb$(s(n),a(n))$ as the class of languages accepted by
$s(n)$ space bounded nondeterministic machines for which the number of
paths from the start configuration to any configuration is at most
$a(n)$. They show that \nspace-\amb$(s(n),a(n))\subseteq \uspace(s(n)\log a(n))$ (hence \nspace-\amb$(\log n,
O(1))\subseteq \UL$). Our method can be used to show that
\nspace-\amb$(s(n),a(n))\subseteq \uspace(s(n)+\log a(n))$,
thus substantially improving their upper bound.

We extend our first result to show that in fact we can count the
number of accepting paths of a $\ReachFewL$ computation using an
oracle in $\UL\cap\coUL$ and this implies that $\ReachLFew\subseteq
\UL\cap\coUL$ (\ReachLFew\ is similar to the class
\Few\ \cite{CaiHemachandra90} in the polynomial-time setting).

\subsubsection*{Complexity of Min-uniqueness}
Our second consideration is the notion of {\em min-uniqueness} which
is a central notion in the study of unambiguity in the logspace
setting. Min-uniqueness was first used by Wigderson to show that
$\NL\subseteq \oplus\L$ non-uniformly \cite{Wigderson94}. For a
directed graph $G$ and two nodes $s$ and $t$, $G$ is called $st$-{\em
  min-unique} if the minimum length $s$ to $t$ path is unique (if it
exists). $G$ is min-unique with respect to $s$, if it is
$sv$-min-unique for all vertices $v$. While $st$-min-uniqueness was
sufficient for Wigderson's result, Reinhardt and Allender used the
stronger version of min-uniqueness to show that $\NL\subseteq
\UL/\poly$. In particular, they essentially showed that a logspace
algorithm that transforms a directed graph into a min-unique graph
with respect to the start vertex can be used to design an unambiguous
algorithm for reachability.  This technique was subsequently used in
~\cite{BTV09} to show that reachability for planar directed graphs is
in $\UL$. These results strongly indicate that understanding
min-uniqueness is crucial to resolving the $\NL$ vs $\UL$ problem.

Our second set of results is aimed at understanding min-uniqueness
from a complexity-theoretic point of view. First we observe that
min-uniqueness is necessary to show that $\NL=\UL$: if $\NL=\UL$, then
there is a $\UL$ algorithm that makes any directed graph min-unique
with respect to the start vertex. It is an easy observation that
Reinhardt and Allender's technique will work even if the algorithm
that makes a directed graph min-unique is only $\UL$ computable. Thus
min-uniqueness is necessary and sufficient for showing $\NL =\UL$.

\vspace{3mm}
\noindent {\em Result 2}: $\NL=\UL$ if and only if there is a
polynomially-bounded $\UL$-computable weight function $f$ so that for
any directed acyclic graphs $G$, $f(G)$ is min-unique with respect to
$s$.
\vspace{2mm}

Graph reachability problems and logspace computations are
fundamentally related. While, reachability in directed graphs
characterizes \NL, Reingold's break-through results implies that
reachability in undirected graphs captures \L\ \cite{Reingold08}.  We
ask the following question. Can we investigate the notion of
min-uniqueness in the context of complexity classes? We introduce a
logspace function class $\UOptL[\log n]$ towards this goal.

\OptL\ is the function class defined by \`{A}lvarez and Jenner (in
\cite{AJ93}) as the logpsace analog of Krentel's \OptP\
\cite{Krentel88}. \OptL\ is the class of functions whose values are
the maximum over all the outputs of an \NL-transducer. \`{A}lvarez and
Jenner showed that this class captures the complexity of some natural
optimization problems in the logspace setting (eg. computing the
lexicographically maximum path of length $\leq n$ from $s$ to $t$ in a
directed graph).

We consider $\OptL[\log n]$, the restriction of $\OptL$ where the
function values are bounded by a polynomial. \`{A}lvarez and Jenner
considered this restriction and showed that $\OptL[\log n] =
\FL^{\NL}[\log n]$. However,  previously there were no completeness
results known for this class.  We show the first completeness result
for $\OptL[\log n]$.  Consider the problem: Given $G$ and two nodes
$s$ and $t$. Compute the length of the shortest path from $s$ to $t$
(denoted by {\sc ShortestPathLength}).  We show that {\sc
  ShortestPathLength} is complete for the class $\OptL[\log n]$ (under
metric reductions).

\vspace{3mm}
\noindent{\em Result 3}. {\sc ShortestPathLength} is complete for $\OptL[\log n] = \FL^{\NL}[\log n].$ 
\vspace{2mm}

Motivated by this completeness result, we define a new unambiguous
function class $\UOptL[\log n]$ (unambiguous $\OptL$: the minimum is
output on a unique computation path).  We show that $\NL = \UL$ is
equivalent to to the question whether $\OptL[\log n]=\UOptL[\log n]$.

\vspace{3mm}
\noindent{\em Result 4}.  $\NL = \UL$ if and only if $\OptL[\log
  n]=\UOptL[\log n]$.
\vspace{2mm}

\SPL, the `gap' version of \UL, is an interesting logspace class first
studied in \cite{AllenderMatching99}. The authors showed that the
`matching problem' is contained in a non-uniform version of \SPL. They
also show that \SPL\ is powerful enough to contain \FewL.  We show
that $\UOptL[\log n] \subseteq \FL^{\SPL}[\log n]$. Thus any language that is
reducible to $\UOptL[\log n]$ is in the complexity class $\SPL$. This
contrasts with the equivalence $\OptL[\log n] = \FL^{\NL}[\log n]$.
We also show that the class $\LogFew$ reduces to $\UOptL[\log n]$ (refer
to the next section for the definition of \LogFew).

\vspace{3mm}
\noindent{\em Result 5}.  $\LogFew \leq \UOptL[\log n] \subseteq \FL^{\SPL}[\log n]$.
\vspace{2mm}

Figures 1 and 2 depict the relations among various unambiguous and
`few' classes known before and new relations that we establish in this
paper, respectively. Definitions of these complexity classes are given
in subsequent sections.

\subsubsection*{Three pages are sufficient for \NL}

Finally we consider the reachability problem for directed graphs
embedded on 3 pages and show that it is complete for \NL. This is in
contrast with reachability for graphs on 2 pages which is logspace
equivalent to reachability in grid graphs and hence is in $\UL$ by the
result of \cite{BTV09}. Thus in order to show that $\NL = \UL$, it is
sufficient to extend the results of \cite{BTV09} to graphs on 3
pages. It is also interesting to note that reachability for graphs on
1 page is equivalent to reachability in trees and is complete for
$\L$.

\vspace{3mm}
\noindent{\em Result 6}.  Reachability in directed graphs embedded on 3 pages is complete for $\NL$.
\vspace{2mm}

We use a combination of existing techniques for proving our results.
\comment{ In particular, techniques that we use can be found in
  ~\cite{FKS84, AllenderMatching99, BJLR91, BDHM92, AJ93} or else
  were.}

%% file: draft.tex
\begin{figure}[t]
\begin{minipage}[b]{0.5\linewidth} 
\centering
\includegraphics[scale=0.75]{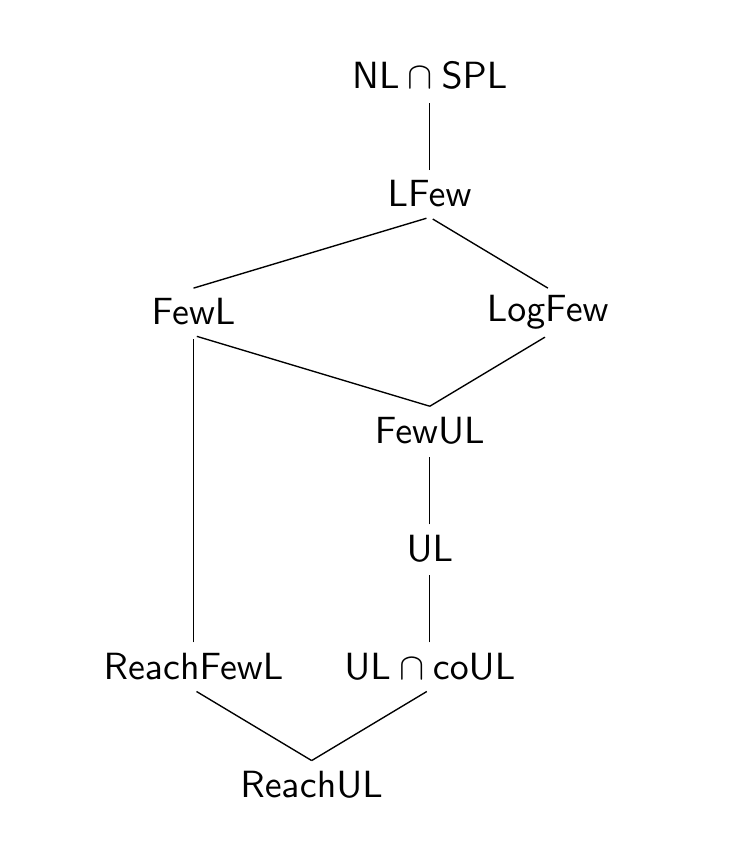}
\caption{Relations known before.}
\label{fig:old-relations}
\end{minipage}
\hspace{0.5cm} 
\begin{minipage}[b]{0.5\linewidth}
\centering
\includegraphics[scale=0.75]{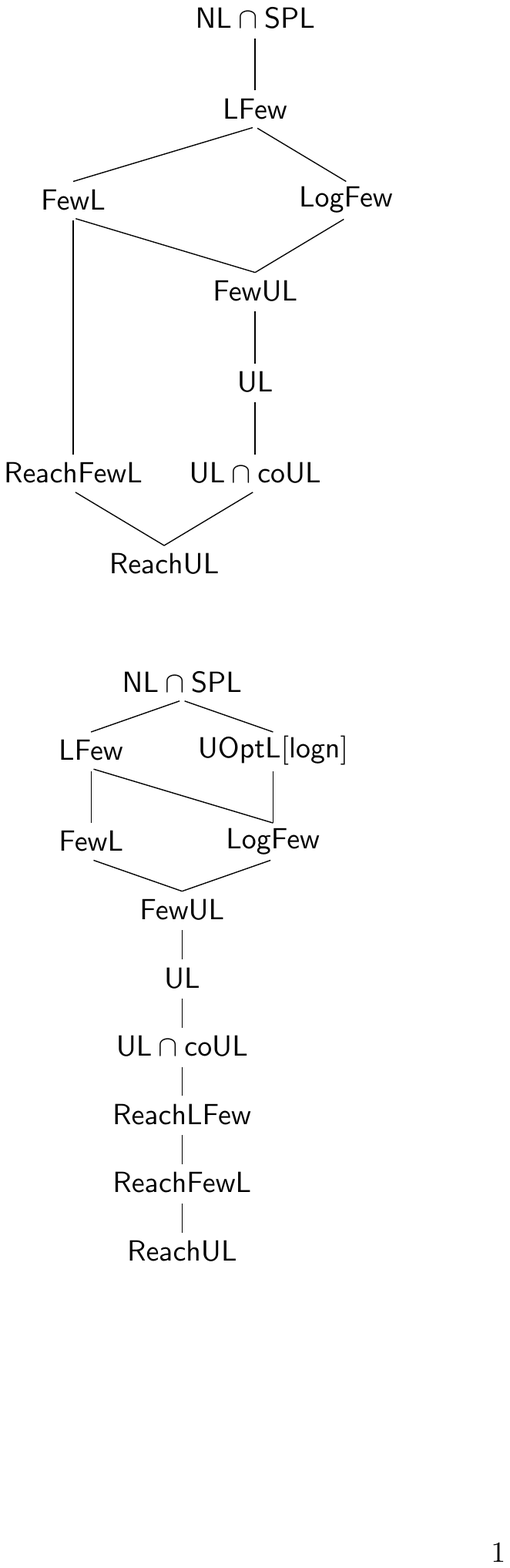}
\caption{New relations.}
\label{fig:new-relations}
\end{minipage}
\end{figure}



\section{Logspace Complexity Classes}

We assume familiarity with the basics of complexity theory and
in particular the log-space bounded complexity class $\NL$. 
It is well known that checking for
$st$-connectivity for general directed graphs is
$\NL$-complete. We call a nondeterministic logspace machine an \NL\ machine. For an $\NL$ machine $M$, let $\acc_{M}(x)$ and $\rej_{M}(x)$ denote the number of accepting computations and the number of rejecting computations respectively. Denote $\gap_{M}(x) = \acc_{M}(x)-\rej_{M}(x)$.  

We are interested in various restrictions of $\NL$ machines with few accepting paths. In the literature (eg \cite{BJLR91,BDHM92,AJ93,AllenderMatching99}) various versions of unambiguity and fewness have been studied. We first define them all here. 

\begin{definition} ({\bf Unambiguous machines})
A nondeterministic logspace machine $M$ is 
\begin{itemize}
\vspace{-2mm}
\item[-]{\em reach-unambiguous} if for any input and for any
  configuration $c$, there is at most one path from the start
  configuration to $c$. (The prefix `reach' in the term indicates that
  the property should hold for all configurations reachable from the
  start configuration).
\vspace{-2mm}
\item[-] {\em unambiguous} if for any input there is at most one accepting path. 
\vspace{-2mm}
\item[-] {\em weakly unambiguous} if for any accepting configuration $c$ there is at most one path from the start configuration to $c$. 
\end{itemize}
\end{definition}
\begin{definition} ({\bf Unambiguous classes})
\begin{enumerate}
\item[-] $\ReachUL$ - class of languages that are decided by reach-unambiguous machines with at most one accepting path on any input. 
\item[-] $\UL$ - class of languages that are decided by unambiguous machines. 
\item[-] $\FewUL$ - 
class of languages that are decided by weakly unambiguous machines. 
\item[-] $\LogFew$ - class of languages $L$ for which there exists a weakly unambiguous machine $M$ and a logspace computable predicate $R$ such that $x\in L$ if and only if $R(x,\acc_{M}(x))$ is true. 
\end{enumerate}
\end{definition}

We could define a `reach' version of $\FewUL$.  
But that coincides with 
$\ReachUL$ as shown in \cite{BJLR91}. The following containments are easy: $\ReachUL
\subseteq \UL\subseteq \FewUL\subseteq \LogFew$. It is also known that 
$\FewUL$ is $\L_{d}(\UL)$ (logspace disjunctive truth-table closure of 
$\UL)$ \cite{BJLR91}. 
  
By relaxing the unambiguity condition to a polynomial bound on the number of paths, we  get analogous `few' classes.  

\begin{definition} ({\bf Few machines})
A nondeterministic logspace machine $M$ is a 
\begin{itemize}
\vspace{-2mm}
\item[-]{\em reach-few} machine if there is a polynomial $p$ so that
  for any input $x$ and for any configuration $c$, there are at most
  $p(|x|)$ paths from the start configuration to $c$.
\vspace{-2mm}
\item[-]{\em few} machine if there is a polynomial $p$ so that for any
  input $x$ there are at most $p(|x|)$ accepting path.
\end{itemize}
\end{definition}

\begin{definition} ({\bf Few classes})
\begin{enumerate}
\item[-] $\ReachFewL$ - class of languages that are decided by reach-few machines. 
\item[-] $\ReachLFew$ - class of languages $L$ for which there exists a reach-few machine $M$ and a logspace computable predicate $R$ such that $x\in L$ if and only if $R(x,\acc_{M}(x))$ is true. 
\item[-] $\FewL$ - class of languages that are decided by few-machines. 
\item[-] $\LFew$ - class of languages $L$ for which there exists a few machine $M$ and a logspace computable predicate $R$ such that $x\in L$ if and only if $R(x,\acc_{M}(x))$ is true. 
\end{enumerate}
\end{definition}

As mentioned in the introduction, $\ReachFewL$ is the same class as $\nspace-\amb(\log n, n^{O(1)})$ defined in \cite{BHS93}. In \cite{BJLR91}, the authors observe that $\ReachFewL \subseteq \LogDCFL$. This is  because a depth first search of a reach-few machine can be implemented in \LogDCFL. 

The following containments follow from the definitions: $\ReachFewL\subseteq \FewL\subseteq \LFew$. It is also clear that all the above-defined classes are contained in $\LFew$ and it is shown in $\cite{AllenderMatching99}$ that $\LFew\subseteq \NL$. Thus all these classes are contained in $\NL$. Finally, we also consider the class \SPL\ - the `gap' version of $\UL$. A language $L$ is in $\SPL$ if there exists an \NL-machine $M$ so that for all inputs $x$, $\gap_{M}(x)\in \{0,1\}$ and $x\in L$ if and only if $\gap_{M}(x) = 1$. \SPL\ is contained in $\oplus\L$ (in fact all `mod' classes) and it is big enough to contain \LFew \cite{AllenderMatching99}. A nonuniform version of \SPL\ contains the matching problem \cite{AllenderMatching99}. 

We will use {\em metric reductions} for functional reducibility.  A
function $f$ is logspace metric reducible to function $g$, if there
are logsapce computable functions $h_1$ and $h_2$ so that $f(x) =
h_1(x,g(h_2(x)))$.

\section{$\ReachFewL \subseteq \UL\cap\coUL$}


We will use the technique of Reinhardt and Allender to show the upper bound. We will state their theorem in a suitable form.  But first we repeat the definition of min-uniqueness.  

\begin{definition}
Let $G = (V,E)$ be a directed graph. For a pair of vertices $s$ and $t$ we say $G$ is $st$-{\em min-unique} if there is a path from $s$ to $t$ in $G$, then the minimum length path from $s$ to $t$ is unique. $G$ is called {\em min-unique} with respect to vertex $s$, if for
all vertices $v$, $G$ is $sv$-min-unique. 
$G$ is called {\em min-unique}
if it is min-unique
with respect to all the nodes. 
\end{definition}

The following theorem from \cite{ReinhardtAllender00} states that 
the reachability problem can be solved unambiguously for classes of graphs that are min-unique with respect to the start vertex. Moreover, we can also check whether a graph is min-unique unambiguously. 
\newpage
\begin{theorem}[\cite{ReinhardtAllender00}]\label{RA}
There is an unambiguous nondeterministic logspace machine $M$ that on input a directed graph $G$ and two vertices $s$ and $t$ such that 
\begin{enumerate}
\item If $G$ is not min-unique with respect to $s$, then $M$ outputs `not min-unique' on a unique path.
\item If $G$ is min-unique with respect to $s$, then $M$ accepts on a
  unique path if there is a directed path
  from $s$ to $t$, and rejects on a unique path if there are no paths from $s$ to $t$.
 \end{enumerate}
\end{theorem}

We can also define the notion of min-uniqueness for weighted
graphs. But this is equivalent to the above definition for our
purposes if the weights are positive and polynomially bounded as we
can replace an edge with weight $k$ with a path of length $k$. In fact
we will some times use this definition for weighted graphs without
explicitly mentioning it.  Thus for showing that $\NL=\UL$ it is
sufficient to come up with a positive and polynomially bounded weight
function that is \UL-computable and makes a directed graph min-unique
with respect to the start vertex.

\comment{
\subsection{Complexity of \ReachFewL}

\begin{definition}
$\ReachFewL$  is the class of languages for which there is a
nondeterministic logspace machine which satisfies for any input $x$ and
any configuration $c$ that there are at most polynomially many paths
from start configuration to $c$.
\end{definition}

\begin{definition}
$\ReachLFew$ is the class of languages $L$ for which there exists an $\NL$ machine $M$ and a logspace computable predicate $R$ so that (a) for any input $x$ and any configuration $c$ of $G_x$, the number of paths from the start configuration $s$ to $c$ is bounded by a polynomial (b) $x\in L$ if and only if $R(x,\#{\it acc}_M(x))$ is true.  
\end{definition}

Clearly $\ReachFewL\subseteq \ReachLFew$. 
We will show that $\ReachLFew\subseteq \UL\cap\coUL$. For that we first need to show that $\ReachFewL\subseteq \UL\cap\coUL$.  
}

\begin{theorem}
$\ReachFewL \subseteq \UL\cap\coUL$
\end{theorem}

\begin{proof}
Let $L$ be in $\ReachFewL$ decided by the machine $M$. Let $G_{(M,x)}$
be the configuration graph of $M$ on input $x$ and $s$ be the start
configuration. Let $t$ be the polynomial that bounds the number of
paths from $s$ to any configuration. Consider the edges in the
lexicographical order. For the $i^{th}$ edge give a weight
$2^{i}$. This is a very good weight function that assigns every path
with unique weight. The problem is that this is not polynomially
bounded.  From this weight function we will give a polynomial number
of weight functions that are logspace computable and polynomially
bounded so that for one of them $G_{(M,x)}$ will be min-unique with
respect to $s$. Since by Theorem~\ref{RA} it is possible to check
whether a given weight function makes the graph min-unique using a
$\UL\cap\coUL$ computation, we can go through each weight function
sequentially.

We will use the well known hashing technique introduced in
\cite{FKS84} for making the graph min-unique.  Let $N$ be the total
number of configurations of $M(x)$. With respect to the above
mentioned weight function, the weight of any path is bounded by
$2^{N+1}$. Let $p_1, p_2, \ldots, p_l$ be the first $l$ distinct prime
numbers so that $\prod_{i=1}^{l}p_i > 2^{N+1}t^2(N)$. Then $l\leq N^5$
and $p_l \leq N^6$. Hence each $p_i$ has a logarithmic bit
representation.

Let {\cal P} be the set of all paths from $s$ and $w_i$ be the weight of the $i^{th}$ path in {\cal P}. Consider the product
$\prod_{i,j}(w_i-w_j)$. This product is bounded by $2^{N+1}t^2(N)$ and is nonzero since for any pair $i,j$ such that $i\neq j$, $w_i\neq w_j$.  
Thus $\prod_{i,j}(w_i-w_j)\neq 0 (\mod \prod p_i)$.  
Hence there should be one (first) $p_k$ with respect to which the product 
is non-zero and modulo this $p_k$, $w_i \neq w_j$ for all $i,j$. That is the weight function $w \mod p_k$ is a weight
function which is $\UL$-computable for
which the configuration graph is min-unique with respect to the start
configuration (\UL-computable because, by Theorem \ref{RA}, we can go through each prime and reject those which are not `good' using a \UL\ computation, until we reach $p_k)$. 
\end{proof}

Buntrock, Hemachandra,
and Siefkes \cite{BHS93} defined, for a space bound $s$
and an unambiguity parameter $a$, the class
\nspace-\amb$(s(n),a(n))$ as the class of languages accepted by
$s(n)$ space bounded nondeterministic machines for which the number of
paths from the start configuration to any configuration is at most
$a(n)$. As one of their main theorems, the authors showed that 
\nspace-\amb$(s(n),a(n))\subseteq \uspace(s(n)\log a(n))\
$ (hence \nspace-\amb$(\log n,                                                  
O(1))\subseteq \UL$). Our method can be used to show that
\nspace-\amb$(s(n),a(n))\subseteq \uspace(s(n)+\log a(n))$,
thus substantially improving their upper bound.

\begin{theorem}
For a space bound $s(n)\geq \log n$ and ambiguity parameter $a(n)$ computable in space $s(n)$ so that $a(n)=2^{O(s(n))}$, \nspace-\amb$(s(n),a(n))\subseteq \uspace(s(n)+\log a(n))$.  
\end{theorem}

\begin{theorem}\label{reachlfew}
Let $L\in \ReachFewL$ accepted by a reach-few machine $M$. Then the $\# L$ function $\acc_{M}(x)$ is computable in $\FL^{\UL\cap\coUL}$. 
\end{theorem}

\comment{
Proof is given in the Appendix.
}

\begin{proof} 
The idea is to compute the number of paths from $s$ to $t$ of a
\ReachFewL-computation with queries to $\UL\cap\coUL$ language using a
logspace machine. If we make sure that all paths from $s$ to $t$ are
of different weights then we can count them by making queries of the
form ``is there a path of length $i$ from $s$ to $t$'' for all $i\leq
N$ and by counting the number of positive answers.

We will use primes as before. But among polynomially many primes we
have to reject those primes that does not give distinct weights to
paths from $s$ to $t$. Notice that Theorem~\ref{RA} can only be used
to rejects primes that do not make the graphs min-unique. It is
possible that some prime makes the graph min-unique with respect to
$s$ but the graph may still have two paths from $s$ to $t$ of the same
weight. For checking this more strict condition, we use the above
result that $\ReachFewL$ is in $\UL\cap\coUL$.

Let $L$ be a language in $\ReachLFew$ witnessed by a machine $M$ and a
polynomial $q$ so that for every $x$, the number of paths from the
start configuration of $M(x)$ to any configuration $c$ is bounded
$q(|x|)$. Let $G_{(M,x)}$ denote the standard layered configuration
graph of $M(x)$. Then this graph also satisfy the property that the
number of paths from the start configuration in the first layer to any
configuration $c$ is bounded by $q(|x|)$. Then the following language
is in $\UL\cap\coUL$: $L = \{(x,c,i)\mid$ there is a path of length
$i$ from $s$ to $c$ in $G_{(M,x)}\}$.

In order to check whether $p$ is a `bad' prime, we need to check
whether there are two paths from $s$ to $t$ of the same weight.

\noindent
``$p$ \mbox{ is bad } $\Leftrightarrow \exists w \exists e =(c,c') \exists a \exists$ a path of length $a$ from $s$ to $c \wedge \exists$ a path of weight $w-w(e)-a$ from $c'$ to $t$  $\wedge \exists$ a path of length $w$ from $s$ to $t$ in $G-{e}$'' 

This can be decided with polynomially many queries to $L$. Once we get a good prime $p$, we can use $L$ as oracle to count the number of distinct paths from $s$ to $t$ using a deterministic logspace machine. This gives $\ReachLFew \subseteq \UL\cap\coUL$.

\end{proof}
\comment{
\begin{definition}
A vertex $v$ in $G$ is said to be {\em unconjusted} if there are only polynomially many paths from $s$ to $v$ and from $v$ to $t$. 

A graph $G$ is said to be {\em uncongested} if it contains a vertex $v$ that is on some path from $s$ to $v$ that is unconjusted vertex.

Define $\uncongested = \{\langle G,s,t \rangle | G \textrm{ is uncongested} \}$.
\end{definition}

\begin{theorem}
$\uncongested \in \SPL$.
\end{theorem}
}

\begin{corollary}
$\ReachLFew \subseteq \UL\cap\coUL$
\end{corollary}

\section{Complexity of Min-uniqueness}

Theorem~\ref{RA} states that min-uniqueness is sufficient for showing $
\NL =\UL$. Next we prove that if $\NL=\UL$ then there is a $\UL$-computable 
weight function that makes any directed acyclic graph min-unique with 
respect to the start vertex. Thus min-uniqueness is necessary and sufficient 
for showing $\NL = \UL$.

\begin{theorem}
$\NL=\UL$ if and only if there is a polynomially-bounded
  $\UL$-computable weight function $f$ so that for any directed
  acyclic graphs $G$, $f(G)$ is min-unique with respect to $s$.
\end{theorem}
\begin{proof}
The reverse direction follows from the above theorem due to Reinhardt
and Allender. For the other direction the idea is to compute a
spanning tree of $G$ rooted at $s$ using reachability queries.  Since
$\NL$ is closed under complement, under the assumption that $\NL=\UL$,
reachability is in $\UL\cap\coUL$. Thus the following language $A =
\{(G,s,v,k) \mid$ there is a path from $s$ to $v$ of length $\leq k
\}$ is in $\UL\cap\coUL$.

The tree can be described as follows.  We say that a vertex $v$ is in
level $k$ if the minimum length path from $s$ to $v$ is of length
$k$. A directed edge $(u,v)$ is in the tree if for some $k$ (1) $v$ is
in level $k$ (2) $u$ is the lexicographically first vertex in level
$k-1$ so that $(u,v)$ is an edge.

It is clear that this is indeed a well defined tree and deciding
whether an edge $e=(u,v)$ is in this tree is in $\L^{A} \subseteq
{\UL\cap\coUL}$.

Now for each edge in the tree give a weight 1. For the rest of the
edges give a weight $n^2$. It is clear that shortest path from a vertex
with respect to this weight function is min-unique with respect to $s$
and it is computable using a $\UL$-transducer.
 
\end{proof}

\`{A}lvarez and Jenner~\cite{AJ93} defines $\OptL$ as the logspace
analog of Krental's \OptP. They show that $\OptL$ captures the
complexity of some natural optimization problems in the logspace
setting (eg. computing lexicographically maximum path of length $\leq
n$ from $s$ to $t$ in a directed graph). They also consider
$\OptL[\log n]$ where the function values are bounded by a polynomial
(hence has $O(\log n)$ bits representations).  Here we revisit the
class \OptL\ \cite{AJ93} and study them in relation to the notion of
min-uniqueness. We define $\OptL$ as a minimization class and show
that computing the minimum length path from $s$ to $t$ in a directed
graph is complete (under metric reductions) for $\OptL[\log n]$.
 
\begin{definition}
An $\NL$-transducer is a nondeterministic logspace bounded Turing
machine with a one-way output tape in addition to its read-only input
tape and read/write work tapes. We will assume that an
$\NL$-transducer will not repeat any configuration during its
computation. Hence its configuration graph contains no cycles and all
computation paths will halt with accepting or rejecting state after
polynomially many steps. Let $M$ be such a $\NL$-transducer.  An
output on a computation path of $M$ is valid if it halts in an
accepting state. For any input $x$, {\em opt}$_{M}(x)$ is the minimum
value over all valid outputs of $M$ on $x$. If all the paths reject,
then {\em opt}$_M(x) = \infty$.  Further, $M$ is called {\em
  min-unique} if for all $x$ either $M(x)$ rejects on all paths or
$M(x)$ outputs the minimum value on a unique path.
\end{definition}

\begin{definition}
A function $f$ is in \OptL\ if there exists a $\NL$-transducer $M$ so that for any $x$, $f(x) = \mbox{{\em opt}}_{M}(x)$. 
A function $f$ is in $\UOptL$ if there is a min-unique nondeterministic transducer $M$ so that for any $x$, $f(x)=\mbox{\it opt}_{M}(x)$. Define $\OptL[\log n]$ and $\UOptL[\log n]$ as the restriction of $\OptL$ and $\UOptL$ where the output of the transducers are bounded by $O(\log n)$ bits.
\end{definition}

If the output is unrestricted, then the computation path of an $\NL$-transducer can be encoded in the output and hence all the output can be made distinct. Hence the classes $\OptL$ and $\UOptL$ are equivalent. But if we restrict the output to be of $O(\log n)$ bits the classes $\OptL$ and $\UOptL$ coincide if and only if $\NL=\UL$ as we show next. 

We will need the following proposition shown in \cite{AJ93}. 
$\FL^{\NL}[\log n]$ denotes the subclass of $\FL^{\NL}$ where the output 
length is bounded by $O(\log n)$. 

\begin{proposition}[\cite{AJ93}]
$\OptL[\log n]= \FL^{\NL}[\log n]$. 
\end{proposition}
  
\begin{theorem}
$\OptL[\log n]=\UOptL[\log n]$ if and only if $\NL = \UL$.
\end{theorem}
\begin{proof}
$\NL=\UL \Rightarrow \OptL[\log n]= \UOptL[\log n]$: Since $\NL$ is
  closed under complement, if $\NL=\UL$ then $\NL=
  \UL\cap\coUL$. Hence $\OptL[\log n] = \FL^{\NL}=
  \FL^{\UL\cap\coUL}$. For a function $f\in \OptL$, let $M$ be $\FL$
  machine that makes query to a language $L\in \UL\cap\coUL$ and
  computes $f$. Let $N$ be the unambiguous machine that decided
  $L$. The {min-unique} transducer $M'$ will simulate $M$ and whenever
  a query $y$ is made to $L$, it will simulate $N$ on $y$ and continue
  only on the unique path where it has an answer. In the end $M'$ will
  output the value computed by $M$ on a unique path.

$\OptL[\log n] = \UOptL[\log n] \Rightarrow \NL = \UL$: Let $L\in
  \NL$. Since $\NL$ is closed under complement, there is a
  nondeterministic machine $M$ that on input $x$ accepts on some path
  and outputs `?' on all other paths if $x\in L$, and rejects on some
  paths and outputs `?' on all other paths if $x\not\in L$. We will
  show that under the assumption $L\in \coUL$. Consider the
  $\NL$-transducer which on input $x$ simulates $M(x)$ and outputs 1
  if $M$ accepts and outputs $0$ if $M$ rejects and outputs a large
  value on paths with `?'. Let $N$ be min-unique machine that computes
  this $\OptL$ function. Thus if $x\not\in L$ then $N(x)$ has a unique
  path on which it outputs 0 (and there may be paths on which it
  outputs 1). If $x\in L$ then there is no path it outputs 0. Now
  consider the machine $N'$ that simulates $N$ and if $N$ outputs 0
  then it accepts. For all other values $N'$ rejects. Clearly this is
  an unambiguous machine that decides $\overline{L}$.

\end{proof}

Next we will exhibit a natural problem that is complete for $\OptL[\log n]$. 
Consider the computational problem {\sc ShortestPathLength}

\begin{itemize}
\item[-] {\sc ShortestPathLength}: Given $(G,s,t)$ where $G = (V,E)$ is a directed graph and $s$ and $t$ are two vertices in $V$.
Compute the length of the shortest path from $s$ to $t$. If no path exists then output $\infty$. 
\end{itemize}

\begin{theorem}
{\sc ShortestPathLength} is complete for $\OptL[\log n]$ (under metric reductions)
\end{theorem}
\begin{proof}
For the containment in $\OptL[\log n]$, consider the $\NL$-transducer, which 
guesses a path of length $\leq n$ from $s$ to $t$. It the guess succeeds 
then outputs the length of the path. Else it rejects. If $G$ has a path from 
$s$ to $t$, then the best path will be of length $\leq n$ hence the minimum 
among the outputs will be the length of the best path. 
 
For the completeness, let $f$ be a function in $\OptL[\log n]$ computed by 
an $\NL$-transducer $M$. Since the output of $M$ is of length $c\log n$ for 
some constant $c$, we will assume that $M$ stores the intermediate value of 
the output on a separate work-tape (called the output work-tape) until the 
end of the computation, and before halting, $M$ copies the contents of this 
work tape to the output tape deterministically and halts. Thus the 
configuration of this machine will also include the contend of this output 
work-tape. We will denote a typical configuration by the tuple $(c,o)$ where 
$o$ is the content of the output work tape. 
 We will assume that at the start configuration the contents of this work-tape is 0.

Consider the following layered weighted graph $G_{(M,x)}$. $G_{(M,x)}$ has $p(|x|)+1$ layers were $p$ is the polynomial bounding the running time of $M$. For $1\leq i \leq p(|x|)$, the $i^{th}$ layer has vertices $(i,c,o)$ where $(c,o)$ is a configuration. The last layer which has just one vertex $t$. 
There is an edge from $(i,c,o)$ to $(i+1, c', o')$ if there is a valid move from the configuration $(c,o)$ to $(c',o')$. The weight of this edge is $(o'- o) + n^{k}$ where $k$ is a large constant so that $n^k > p(n)\times n^c $. We will also add edges from $(i,c,o)$ to $(i+1,c,o)$ if $(c,o)$ is an accepting configuration. The weight of this edge is $n^k$. Finally we will add an edge with weight $n^k$ from $(p(n),c,o)$ to $t$ if $(c,o)$ is an accepting configuration. For correctness, any computation path of $M$ with an output $o$ corresponds to a path in $G_{(M,x)}$ from the start configuration to $t$ of weight $o+p(n)n^k$.  Since the weights on the edges are positive and bounded by a polynomial, it is easy to replace to each edge with weight $l$ with a path of length $l$.   
\end{proof}

It can be verified that the standard reductions from directed graph
reachability to other \NL-complete problems also shows that a version
of their optimization problems are $\OptL[\log n]$ complete. For
example {\sc DFAShortestWordLength} 
(Given a DFA $M$. Find the length
of the shortest word that $M$ accepts if $L(M)$ is nonempty) and 
{\sc
  WordGenLength} (Given a set $X$ with an associative binary
operation, a subset $S\subseteq X$, and a word $w$ over $X$. Find the
length of the shortest generation sequence of $w$) are complete for
$\OptL[\log n]$.

\comment{
\begin{theorem} The following problems are complete for $\OptL[\log n]$.
\begin{itemize}
\item[-] {\sc DFAShortestWordLength} Given a DFA $M$. Find the length
  of the shortest word that $M$ accepts if $L(M)$ is nonempty) and

\item[-]

{\sc WordGenLength}: Given a set $X$ with an associative binary operation, a subset $S\subseteq X$, and a word $w$ over $X$. Find the length of the shortest 
generation sequence of $w$.
\end{itemize}
\end{theorem}
}

As $\UOptL[\log n]\subseteq \OptL[\log n]$, $\UOptL[\log n]$ is in 
$\FL^{\NL}[\log n]$. Here we show that $\UOptL[\log n]$ can be computed using a $\SPL$ oracle. Thus if $\NL$ reduces to $\UOptL[\log n]$, then $\NL\subseteq \SPL$.

\begin{theorem}\label{uoptl}
$\UOptL[\log n] \subseteq \FL^{\SPL}[\log n]$
\end{theorem}

\begin{proof}
Let $f\in \UOptL[\log n]$ and let $M$ be the min-unique $\NL$-transducer that witnesses that $f\in \UOptL[\log n]$ and let $p$ be the polynomial bounding the value of $f$. Consider the following language $L$:

$$ L = \{(x,i)\mid f(x) = i \mbox{ and } i\leq p(|x|)\}. $$ 

We will show that $L\in \SPL$. Then in order to compute $f$ a logspace machine will ask polynomially many queries $(x,i)$ for $1\leq i\leq p(n)$. 
 
Consider the following machine $N$ which behaves as follows: $N$ on input $x$ and $i \leq p(n)$, simulates $M$ on input $x$ and accepts if and only if $M$ halts with an output $\leq i$. Let $g(x,i)$ counts the number of accepting paths of $N$ on input $(x,i)$. Notice that for $i < f(x)$, $g(x,i) = 0$, for $i = f(x)$ then $g(x,i) = 1$, and for $i> f(x)$, $g(x,i)\geq 1$. 

Now consider the $\GapL$ function $h(x,j) = g(x,j)\Pi_{i=1}^{j-1}(1-g(x,i))$. It follows that $h(x,j) = 1 $ exactly when $f(x) = i$. For the rest of $i$, $h(x,j)=0$.  Thus $L\in \SPL$. 
  r
\end{proof}

\comment{
The proof is given in the Appendix. 
}

\begin{corollary}
If $\NL \subseteq \L^{\UOptL[\log n]}$ then $\NL \subseteq  \SPL$.
\end{corollary}

An interesting question is whether $\FewL$ reduces to $\UOptL$. We are not able to show this, but we show that the class $\LogFew$ reduces to $\UOptL$. 

\begin{theorem}\label{logfew}
$\LogFew \leq \UOptL[\log n]$ (under metric reductions)
\end{theorem}

\comment{

The proof is given in the Appendix. 
}

\begin{proof}
Let $L$ be a language in $\LogFew$. Let $M$ be a weakly unambiguous
machine that decided $L$. Consider the $\NL$-transducer $N$ that on
input $x$, computes the number of accepting paths of $M(x)$: $N(x)$
guess a $l$ so that $1\leq l\leq p(n)$ (where $p$ is the polynomial
bounding the number of accepting configurations) and then guess $l$
distinct accepting paths in lexicographically increasing accepting
configurations and accepts and outputs $l$ if all of them
accepts. Clearly $N$ outputs ${\it acc}_{M}(x)$ on exactly one
computation path and all other paths that accepts will have output $<
{\it acc}_{M}(x)$.
\end{proof}

%% file: 3page.tex
\section{Three pages are sufficient for $\NL$}

We show that the reachability problem for directed graphs embedded on
3 pages is complete for $\NL$. It can be shown that the reachability
problem for graphs on 2 pages is equivalent to reachability in grid
graphs and hence is in $\UL$ by the result of \cite{BTV09}. Thus in
order to show that $\NL = \UL$ it is sufficient to extend the
techniques of \cite{BTV09} to graphs on 3 pages. It is also
interesting to note that graphs embedded on 1 page are outer-planar and
hence reachability for directed graphs on 1 page is complete for $\L$
\cite{AllenderEtAl06}.

\begin{definition} 
{\page} is the class of all graphs $G$, that can be embedded on $3$
pages as follows: all vertices of $G$ lie along the spine and the
edges lie on exactly one of the two pages without
intersection. Moreover all edges are directed from top to bottom.
{\pageR} is the language consisting of tuples $(G,s,t)$, such that $G
\in \page$, $s$ and $t$ are two vertices in $G$ and there exists a
path from $s$ to $t$ in $G$.
\end{definition}

\begin{theorem}\label{theorem:3page}
{\pageR} is complete for {\NL}.
\end{theorem}

\begin{proof}
Assume that we are given a topologically sorted DAG $G$, with $(u_1,
u_2, \ldots , u_n)$ being the topological ordering of the vertices of
$G$. We want to decide if there is a path in $G$ from $u_1$ to
$u_n$. We define an ordering on the edges of $G$, say
$\mathcal{E}(G)$. Given two edges $e_1$ and $e_2$, (i) if head of
$e_1$ precedes head of $e_2$, then $e_1$ precedes $e_2$ in the
ordering, (ii) if head of $e_1$ is the same as the head of $e_2$, then
$e_1$ precedes $e_2$ in the ordering if tail of $e_1$ precedes tail of
$e_2$. It is easy to see that $\mathcal{E}(G)$ can be constructed in
logspace given $G$ and in any path from $s$ to $t$, if edge $e_1$
precedes $e_2$, then $e_1$ precedes $e_2$ in $\mathcal{E}(G)$ as
well. Let $m$ be the number edges in $G$.
 
We create $2m$ copies of each vertex in $G$ and let $v_i^j$ denote the
$j$th copy of the vertex $u_i$, for $i \in [n]$ and $j \in [2m]$. We
order the vertices along the spine of $H$ from top to bottom as
follows: \\ $(v_1^1, v_2^1, \ldots, v_n^1, v_n^2, v_{n-1}^2, \ldots ,
v_1^2 ,v_1^3, v_2^3, \ldots, v_n^3, \ldots, v_n^{2m}, \ldots,
v_1^{2m}).$

Next we need to connect all the $2m$ vertices corresponding to each
$u_i$ from the top to bottom. We use the first $2$ pages to do
that. Put the edge $(v_i^j , v_i^{j+1})$ in $H$, for each $i \in [n]$
and each $j \in [2m -1]$, using page $1$ when $j$ is odd and page $2$
when $j$ is even. For the $k$th edge in $\mathcal{E}(G)$, say $e_k =
(u_{k_1}, u_{k_2})$, put the edge $(v_{k_1}^{2k-1} , v_{k_2}^{2k})$ in
$H$, using page $3$. It is clear that this can be done without any two
edges crossing each other. We give an example of this reduction in
Figure \ref{fig:3page_example}. The claim is, there exists a path from
$u_1$ to $u_n$ in $G$ if and only if there exists a path from $v_1^1$
to $v_n^{2m}$ in $H$.

\input{page3_example.tex}

Suppose there exists a path $p$ from $u_1$ to $u_n$ in $G$. Let $p =
(e_{i_1}, \ldots e_{i_l})$. For each $j \in [l]$, corresponding to
$e_{i_j}$ there exists an edge in page $3$ of $H$ by construction, say
$f_{j}$. Also by construction and the ordering $\mathcal{E}(G)$, the
tail of $f_j$ lies above the head of $f_{j+1}$ along the spine of
$H$. Further, since the head of $e_{i_{j+1}}$ is the same as the tail
of $e_{i_j}$ for $j \in [l-1]$, there exists a path from the tail of
$f_j$ to the head of $f_{j+1}$ (using edges from pages $1$ and
$2$). Thus we get a path from $v_1^1$ to $v_n^{2m}$ in $H$.

To see the other direction, let $\rho$ be a path from $v_1^1$ to
$v_n^{2m}$ in $H$. Let $\rho_3 = (\alpha_{1}, \alpha_{2}, \ldots ,
\alpha_{r} )$ be the sequence of edges of $\rho$ that lie on page
$3$. Note that each of the edges in $\rho_3$ has a unique pre-image in
$G$ by the property of the reduction. This defines a sequence of edges
$p'$ in $G$ by taking the respective pre-images of the edges in
$\rho_3$. Now the sub-path of $\rho$ from the $v_1^1$ to the head of
$\alpha_{1}$ uses only edges from page $1$ and $2$ and thus by
construction the head of $\alpha_{1}$ is a vertex $v_1^{l_1}$ (for
some $l_1 \in [2m]$). Similar argument establishes that the tail of
$\alpha_{r}$ is a vertex $v_n^{l_2}$ (for some $l_2 \in [2m]$) and
also that the tail of $\alpha_{i}$ and the head of $\alpha_{i+1}$ are
the copies of the same vertex in $G$, for $i \in [r-1]$. Therefore
$p'$ is a path from $u_1$ to $u_n$ in $G$.
\end{proof}

%% file: page3_example.tex
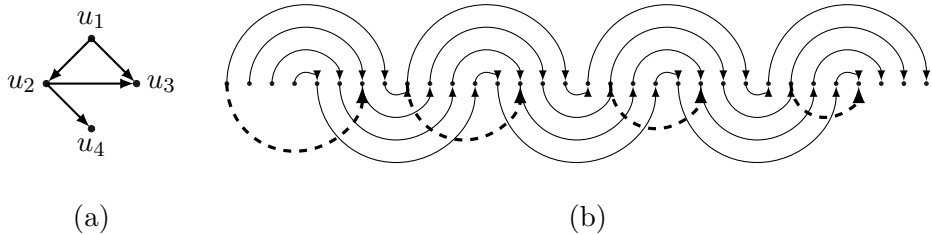
\begin{figure}[h]
\centering

\begin{tikzpicture}[scale=.6,shorten >=.35mm,>=latex]

\draw [fill, color=black!90] (0,0) circle(0.7 mm);
\draw [fill, color=black!90] (2,0) circle(0.7 mm);
\draw [fill, color=black!90] (1,1) circle(0.7 mm);
\draw [fill, color=black!90] (1,-1) circle(0.7 mm);

\draw [thick][->] (1,1) -- (0,0);
\draw [thick][->] (1,1) -- (2,0);
\draw [thick][->] (0,0) -- (2,0);
\draw [thick][->] (0,0) -- (1,-1);

\node [above] at (1,1) {$u_1$};
\node [left] at (0,0) {$u_2$};
\node [right] at (2,0) {$u_3$};
\node [below] at (1,-1) {$u_4$};

\node at (1,-3) {(a)};

\foreach \z in {0}
{
\foreach \y in {4,...,19}
{
\draw [fill, color=black!90] (\y,\z) circle(0.4 mm);
\draw [fill, color=black!90] (\y+ .5,\z) circle(0.4 mm);
}

\foreach \x in {1.75,1.25,0.75,0.25}
{
	\foreach \y in {4,8,12,16}
	{
	\draw [->] (\y + 1.75- \x,\z) arc (180:0:\x);
	}
}
\foreach \x in {1.75,1.25,0.75,0.25}
{
	\foreach \y in {6,10,14}
	{
	\draw [->] (\y + 1.75- \x,\z) arc (180:360:\x);
	}
}

\draw [dashed, very thick, ->] (4,\z) arc (180:360:1.5);
\draw [dashed, very thick, ->] (8,\z) arc (180:360:1.25);
\draw [dashed, very thick, ->] (12.5,\z) arc (180:360:1);
\draw [dashed, very thick, ->] (16.5,\z) arc (180:360:.75);
}
\node at (12,-3) {(b)};

\end{tikzpicture}
\caption{(a) Graph $G$. (b) The corresponding graph $H$. The dashed edges of $H$ are on page $3$.}
\label{fig:3page_example}
\end{figure}

%% file: ack.tex
\section{Acknowledgments}

We thank Eric Allender for an interesting email discussion and
providing valuable suggestions that improved the presentation of the
paper.  We thank V. Arvind for interesting email exchanges on the
topic of this paper. The last author deeply thanks Meena Mahajan and
Thanh Minh Hoang for discussions on a related topic during a recent
Dagstuhl workshop. We thank Samir Datta and Raghav Kulkarni for
discussions which lead to a weaker version of Theorem
\ref{theorem:3page} (namely, reachability for 4-page graphs is
complete for \NL).

\comment{ Theorem \ref{theorem:3page} is obtained in discussion with
  Samir Datta and Raghav Kulkarni and we thank them for letting us
  include it in this paper.  }